\documentclass[11pt]{article}
\usepackage[margin=1in]{geometry}
\usepackage{rpmacros}
\usepackage[T1]{fontenc}
\usepackage[usenames,dvipsnames]{color}
\usepackage{stmaryrd}
\usepackage{lineno}
\usepackage{xfrac}
\usepackage{mathpazo}
\usepackage{todonotes}
\usepackage{setspace}
\usepackage{nicefrac}
\usepackage{gitinfo}
\usepackage{float}
\usepackage{scrextend}
\allowdisplaybreaks

\usepackage{color}
\usepackage[pdfstartview=FitH,pdfpagemode=UseNone,colorlinks=true,citecolor=blue,linkcolor=blue]{hyperref}
\hypersetup{
  linkcolor=[rgb]{0.3,0.3,0.6},
  citecolor=[rgb]{0.2, 0.6, 0.2},
  urlcolor=[rgb]{0.6, 0.2, 0.2}
}
\usepackage[nameinlink,capitalise]{cleveref}

\usepackage{amsthm}
\usepackage{thmtools,thm-restate}

\numberwithin{equation}{section}
\declaretheoremstyle[bodyfont=\it,qed=\qedsymbol]{noproofstyle}

\declaretheorem[numberlike=equation]{observation}

\declaretheorem[name=Observation,numbered=no]{observation*}

\declaretheorem[numberlike=equation]{fact}

\declaretheorem[numberlike=equation]{theorem}

\declaretheorem[name=Theorem,numbered=no]{theorem*}

\declaretheorem[numberlike=equation]{lemma}
\declaretheorem[name=Lemma,numbered=no]{lemma*}
\declaretheorem[numberlike=equation,style=noproofstyle,name=Lemma]{lemmawp}

\declaretheorem[name=Corollary,numbered=no]{corollary*}

\declaretheorem[numberlike=equation]{proposition}
\declaretheorem[name=Proposition,numbered=no]{proposition*}

\declaretheorem[numberlike=equation]{claim}
\declaretheorem[name=Claim,numbered=no]{claim*}

\declaretheorem[numberlike=equation]{conjecture}
\declaretheorem[name=Conjecture,numbered=no]{conjecture*}

\declaretheorem[numberlike=equation]{question}
\declaretheorem[name=Question,numbered=no]{question*}

\declaretheoremstyle[bodyfont=\it,qed=$\lozenge$]{defstyle} 

\declaretheorem[numberlike=equation,style=defstyle]{definition}
\declaretheorem[unnumbered,name=Definition,style=defstyle]{definition*}

\declaretheorem[unnumbered,name=Example,style=defstyle]{example*}

\declaretheorem[unnumbered,name=Notation=defstyle]{notation*}

\declaretheorem[unnumbered,name=Construction,style=defstyle]{construction*}

\declaretheorem[unnumbered,name=Remark,style=defstyle]{remark*}

\usepackage{nth}
\usepackage{intcalc}
\usepackage{etoolbox}
\usepackage{xstring}

\usepackage{ifpdf}
\ifpdf
\else
\usepackage[quadpoints=false]{hypdvips}
\fi

\newcommand{\ECCCurl}[2]{http://eccc.hpi-web.de/report/\ifnumcomp{#1}{>}{93}{19}{20}#1/#2/}

\newcommand{\shortECCC}[2]{\texttt{\href{http://eccc.hpi-web.de/report/\ifnumcomp{#1}{>}{93}{19}{20}#1/#2/}{eccc:TR#1-#2}}}

\newcommand{\parseECCC}[1]{
\StrSubstitute{#1}{TR}{}[\tmpstring]%
\IfSubStr{\tmpstring}{/}{ 
\StrBefore{\tmpstring}{/}[\ecccyear]%
\StrBehind{\tmpstring}{/}[\ecccreport]%
}{
\StrBefore{\tmpstring}{-}[\ecccyear]%
\StrBehind{\tmpstring}{-}[\ecccreport]%
}%
\shortECCC{\ecccyear}{\ecccreport}}

\usetikzlibrary{decorations.pathreplacing,arrows}

\newif\ifnote
\notetrue
\ifnote
\newcommand{\RPnote}[1]{\textcolor{BrickRed}{\guillemotleft RP: #1\guillemotright}}
\newcommand{\MKnote}[1]{\textcolor{Purple}{\guillemotleft MK: #1\guillemotright}}
\newcommand{\ZGnote}[1]{\textcolor{OliveGreen}{\guillemotleft ZG: #1\guillemotright}}
\newcommand{\NSnote}[1]{\textcolor{NavyBlue}{\guillemotleft ZG: #1\guillemotright}}
\else
\newcommand{\RPnote}[1]{}
\newcommand{\MKnote}[1]{}
\newcommand{\ZGnote}[1]{}
\newcommand{\NSnote}[1]{}
\fi

\newcommand{\ehref}[1]{\href{mailto:#1}{#1}}

\newcommand{\D}{\Delta}
\newcommand{\g}{{\mathsf{Gen}}}
\let\epsilon\varepsilon
\newcommand{\coeff}{\operatorname{coeff}}

\newcommand{\Del}{\Delta}

\onehalfspace

\title{Derandomization from Algebraic Hardness\footnote{A preliminary version of this paper appeared in the Proceedings of FOCS 2019~\cite{GKSS19conf}.}
}
\author{%
  Zeyu Guo\thanks{\ehref{zguotcs@gmail.com}. Department of Computer Science, University of Haifa, Israel. This work was done during a postdoctoral stay at IIT Kanpur.}
  \and
  Mrinal Kumar\thanks{\ehref{mrinalkumar08@gmail.com}. Department of Computer Science \& Engineering, IIT Bombay, India.  Part of this work was done while  at the Simons Institute for the Theory of Computing, Berkeley during the semester on Lower Bounds in Computational Complexity in Fall 2018 and during a postdoctoral stay  at the University of Toronto.}
\and
Ramprasad Saptharishi\thanks{\ehref{ramprasad@tifr.res.in}. Tata Institute of Fundamental Research, Mumbai, India. Research supported by Ramanujan Fellowship of DST. }
\and
{Noam Solomon\thanks{ \ehref{noam.solom@gmail.com}. Department of Mathematics, MIT, Cambridge, MA, USA.}}
}

\date{}

\begin{document}

\maketitle
{\let\thefootnote\relax
  \footnotetext{This research was supported in part by the International Centre for Theoretical Sciences (ICTS) during a visit for participating in the program - Workshop on Algebraic Complexity Theory (Code: ICTS/wact2019/03)}
\footnotetext{\textcolor{white}{Base version:~(\gitAuthorIsoDate)\;,\;\gitAbbrevHash\;\; \gitVtag}}
}

\begin{abstract}
  A hitting-set generator (HSG) is a polynomial map $\g:\F^k \to \F^n$ such that for all $n$-variate polynomials $C$ of small enough circuit size and degree, if $C$ is nonzero, then $C\circ \g$ is nonzero.
In this paper, we give a new construction of such an HSG assuming that we have an explicit polynomial of sufficient hardness.
Formally, we prove the following result over any field $\F$ of characteristic zero:
\begin{quote}
Let $k\in \N$ and $\delta > 0$ be arbitrary constants. Suppose $\set{P_d}_{d\in \N}$ is an explicit family of  $k$-variate polynomials such that $\deg P_d = d$ and $P_d$ requires algebraic circuits of size $d^\delta$. Then, there are explicit hitting sets of polynomial size for the class $\VP$.
\end{quote}
This is the first HSG in the algebraic setting that yields a complete derandomization of polynomial identity testing (PIT) for general circuits from a suitable algebraic hardness assumption. Unlike the prior constructions of such maps~\cite{NW94, KI04, AGS19, KST19}, our construction is purely algebraic and does \emph{not} rely on the notion of combinatorial designs.

\medskip

As a direct consequence, we show that even saving \emph{a single point} from the ``trivial'' explicit, exponential sized hitting sets for constant-variate polynomials of low \emph{individual-degree} which are computable by small circuits, implies a deterministic polynomial time algorithm for PIT. More precisely, we show the following:
\begin{quote}
  Let $k\in \N$ and $\delta > 0$ be arbitrary constants. Suppose for every $s$ large enough, there is an explicit hitting set of size at most $((s+1)^k - 1)$ for the class of $k$-variate polynomials of \emph{individual degree} $s$ that are computable by size $s^\delta$ circuits. Then there is an explicit hitting set of size $\poly(s)$ for the class of $s$-variate polynomials, of degree $s$, that are computable by size $s$ circuits. 
\end{quote}
As a consequence, we give a deterministic polynomial time construction of hitting sets for algebraic circuits, if a strengthening of the  $\tau$-Conjecture of Shub and Smale \cite{shub1995,smale98} is true.

\end{abstract}

\section{Introduction}

The interaction of hardness and randomness is one of the most well studied themes in computational complexity theory, and in this work we focus on exploring this interaction further in the realm of algebraic computation.
To set the stage, we start with a brief introduction to algebraic complexity.

The field of algebraic complexity primarily focuses on studying multivariate polynomials and their complexity in terms of the number of basic operations (additions and multiplications) required to compute them.
Algebraic circuits (which are just directed acyclic graphs with leaves labelled by variables or field constant, and internal gates labelled by $+$ or $\times$) form a very natural model of computation in this setting, and the size (number of gates or wires) of the smallest algebraic circuit computing a polynomial gives a robust measure of its complexity.

The main protagonists in the hardness-randomness interaction in algebraic complexity are the \emph{hardness} component, which is the question of proving superpolynomial lower bounds for algebraic circuits for any explicit polynomial family, and the \emph{randomness} component, which is the question of designing efficient \emph{deterministic} algorithms for polynomial identity testing (PIT) --- the algorithmic task of checking if a given circuit computes the zero polynomial.
Both these questions are of fundamental importance in computational complexity and are algebraic analogues of their more well known Boolean counterparts: the $\P \text{ vs } \NP$ question and the $\P \text{ vs } \BPP$ question respectively.
These seemingly different problems are closely related to each other, and in this work we focus on one direction of this relationship; namely, the use of explicit hard polynomial families for derandomization of PIT.

It is known from an influential work of Kabanets and Impagliazzo~\cite{KI04} that lower bounds on the algebraic circuit complexity of explicit polynomial families lead to non-trivial deterministic algorithms for polynomial identity testing (PIT) of algebraic circuits. Moreover, the results in~\cite{KI04} show that stronger lower bounds give faster deterministic algorithms for polynomial identity testing. For instance, from truly exponential (or $2^{\Omega(n)}$) lower bounds, we get quasipolynomial (or $n^{O(\log n)}$) time deterministic algorithms for PIT. From weaker superpolynomial (or $n^{\omega(1)}$) lower bounds, we only seem to get a subexponential (or $2^{n^{o(1)}}$) time PIT algorithm.

However, no matter how good the lower bounds for algebraic circuits are, this connection between lower bounds and derandomization does not seem to give truly polynomial time deterministic algorithms for PIT.
This is different from the Boolean setting, where it is known that strong enough boolean circuit lower bounds imply that $\BPP = \P$~\cite{IW97}.
The difference stems from the fact that, in the worst case, an $n$-variate degree $d$ polynomial $P$ needs to be queried on as many as $\binom{n + d}{d} \gg 2^n$ points to be sure of its nonzeroness.
A key player in this interaction of hardness and randomness, in the context of algebraic complexity, is the notion of a \emph{hitting-set generator} (HSG), which we now define.
\begin{definition}[Hitting-set generators] \label{defn:hsg}
A polynomial map $G:\F^k \to \F^m$ given by $G(z_1, z_2, \ldots, z_k) = \left(g_1(\vecz), g_2(\vecz), \ldots, g_n(\vecz)\right)$ is said to be a hitting-set generator (HSG) for a class ${\cal C} \subseteq \F[x_1, x_2, \ldots, x_n]$ of polynomials if for every nonzero $Q \in {\cal C}$, we have that $Q\circ G = Q(g_1, g_2, \ldots, g_n)$ is also nonzero. 

We shall say that $G$ is \emph{$t(n)$-explicit} if, for any $\veca \in \F^k$ of bit-complexity at most $n$, we can compute $G(\veca)$ in deterministic time $t(n)$. Here $k$ is called the \emph{seed length} of the HSG and $n$ is called the \emph{stretch} of the HSG.
The maximum of the degrees of $g_1, g_2, \ldots, g_n$ is called the \emph{degree} of the HSG.
\end{definition}

If a polynomial map $G$ is an HSG for a class ${\cal C}$ of circuits, we say  $G$ \emph{fools} the class ${\cal C}$.

Informally, an HSG $G$ gives a polynomial map which reduces the number of variables in the polynomials in ${\cal C}$ from $n$ to $k$ while preserving their nonzeroness.
It is not hard to see that such polynomial maps are helpful for deterministic PIT for ${\cal C}$.
To test if a given $n$-variate polynomial $Q \in {\cal C}$ is nonzero, it is sufficient to check that $Q \circ G$, a $k$-variate polynomial, is nonzero.
If the degree of each $g_i$ is not-too-large, then a ``brute-force'' check (via the Ore-DeMillo-Lipton-Schwartz-Zippel Lemma~\cite{O22,DL78,S80,Z79}) can be used to test if $Q\circ G$ is zero in at most $\poly(t(n))\cdot \left(\deg(G) \cdot \deg(Q)\right)^{O(k)}$ time, if $G$ is $t(n)$-explicit.
Thus, it is desirable to have HSGs that are very explicit (small $t(n)$), low degree and large stretch ($k \ll n$).

\paragraph{PIT from the Boolean and algebraic perspective: } 
Typically, in algebraic complexity, the model of algebraic circuits are allowed to introduce arbitrary field constants on the wires, \emph{for free}, to scale the value passed on. The size of an algebraic circuit usually just denotes the number of gates in the circuit (and not the complexity of intermediate constants used). In contrast, when PIT is
interpreted as a Boolean computational task, the circuit is expected to be provided as input which includes the descriptions of any constants used by the input circuit. Hence, the size of the input circuit also includes the complexity of contants on edges of the circuit.

Often, for PIT, it is convenient to construct hitting sets for polynomials that are computable by algebraic circuits of a certain size $s$, which is the number of gates in the circuit and \emph{does not} depend  on the constants used on the edges. As long as the  the hitting set produced is efficiently enumerable by a Turing machine, this is sufficient to solve the Boolean task of PIT as the size of the input is no more than the algebraic size of the circuit.

Throughout this paper, we always deal with this algebraic notion of size which refers to the number of gates in circuit (regardless of the constants used). Therefore, the notion of lower bounds will also refer to the number of gates necessary, and not related to the constants on the wires.  



\subsection{Prior construction of generators}

We shall use $\mathcal{C}(n,d,s)$ to denote the class of $n$-variate polynomials of degree at most $d$ that are computable by size $s$ circuits, and $\mathcal{C}(n,\text{i-deg: }d,s)$ to denote the class of $n$-variate polynomials of \emph{individual degree}\footnote{maximum degree of any variable; for example, a multilinear polynomial has individual degree $1$.} at most $d$ that are computable by size $s$ circuits. 

\paragraph{Generators from combinatorial designs:}
One of the earliest (and most well-known) applications of lower bounds to derandomization is the construction of \emph{pseudorandom generators (PRG)} from hard explicit Boolean functions by  Nisan and Wigderson~\cite{NW94}. In the algebraic setting, an analogous construction was shown to produce HSGs  by Kabanets and Impagliazzo~\cite{KI04}\footnote{Even though the construction of the generator is the same in~\cite{KI04} and~\cite{NW94}, there are crucial differences in the analysis.
  In particular, the analysis for the HSG in~\cite{KI04} relies on a deep result of Kaltofen~\cite{K89} about low degree algebraic circuits being closed under polynomial factorization.}.
These constructions are based on the notion of a combinatorial design, which is a family of subsets that have small pairwise intersection.
Given an explicit construction of such a combinatorial design (e.g.
a family ${\cal F} = \{S_1, S_2, \ldots, S_n\}$ of subsets of $[k]$ of size $t$ each), the PRG/HSG in~\cite{NW94, KI04} is then constructed by just taking a hard polynomial $P(x_1,\ldots, x_t)$ and defining the map as $G(y_1, y_2, \ldots, y_k) = \left(P(\vecy\mid_{S_1}), P(\vecy\mid_{S_2}), \ldots, P(\vecy\mid_{S_n}) \right)$.
The proof of correctness for this HSG goes via a hybrid argument and a result of Kaltofen~\cite{K89}.

\paragraph{Bootstrapping hitting sets and HSGs with large stretch:}
In a recent line of work~\cite{AGS19, KST19} the following surprising \emph{boostrapping} phenomenon was shown to be true for hitting sets for algebraic circuits. The following is the statement from \cite{KST19}:
\begin{theorem}[\cite{KST19}]\label{thm:bs-kst}
  Let $\delta > 0$ and $n \geq 2$ be constants.
Suppose that, for all large enough $s$, there is an explicit hitting set of size $s^{n-\delta}$ for all degree $s$, size $s$ algebraic formulas (or algebraic branching programs, or circuits respectively) over $n$ variables.
Then, there is an explicit hitting set of size $s^{\exp(\exp(O(\log^{*}s)))}$ for the class of degree $s$, size $s$ algebraic formulas (or algebraic branching programs, or circuits respectively) over $s$ variables.
\end{theorem}
In other words, a slightly non-trivial explicit construction of hitting sets even for constant-variate algebraic circuits implies an almost complete derandomization of PIT for algebraic circuits.
A natural question in this direction which has remained  open is the following.
\begin{question}[\cite{KST19}]\label{q:2}
  Can slightly non-trivial hitting sets for constant-variate algebraic circuits can be bootstrapped to get polynomial size (and not just almost polynomial size as in~\autoref{thm:bs-kst}) hitting sets for all circuits ?
\end{question}
The proof of~\autoref{thm:bs-kst} can also be interpreted as a different HSG for algebraic computation.
This HSG, given the hypothesis of \autoref{thm:bs-kst}, stretches $k$ bits to $n$ bits (for arbitrarily large $n$), but the degree and explicitness of the generator grows as $n^{\exp(\exp(O(\log^*n)))}$.
Thus, this construction comes very close to answering~\autoref{q:2} 
without completely answering it. This HSG is essentially constructed via a repeated composition of the HSG in~\cite{KI04, NW94} where, for each step, it uses a different hard polynomial with gradually increasing  hardness. Due to this inherent iterative nature of the construction, it seems difficult to reduce the degree and explicitness of such HSG constructions to $\poly(n)$.

\paragraph{The need to go beyond design-based  HSGs:}
In the set up of boolean computation, observe that we cannot expect to have any PRG (or even HSG) of seed length $k$ to fool circuits of size much larger than $n2^k$ since we can construct a circuit of size $O(n2^k)$ to identify the range of the generator (consisting of $2^k$ strings of length $n$ each). A similar argument gives an upper bound of ${(dD)}^{O(k)}$ on the size of degree $D$ algebraic circuits which can be fooled by a HSG with seed length $k$ and degree $d$. Thus, while the stretch of any boolean PRG constructed via hardness of a boolean function is upper bounded by $n2^k$, in the algebraic setting, one could hope for a construction of hitting-set generators of stretch as large as $d^{k}$ from sufficiently hard explicit polynomial families.\footnote{Indeed, we know from elementary counting (or dimension counting) arguments that there exist degree $d$ polynomials in $k$ variables which require algebraic circuits of size nearly $\binom{d+k}{k}$, which can be approximated by $d^{\Omega(k)}$ when $k$ is much smaller than $d$, which is the range of parameters we work with in this paper.} 
However, till recently, there were no known constructions of  such HSGs with stretch larger than $2^k$ (even for non-constant $k$). Note that obtaining a stretch of $d^k$ in the algebraic world be meaningful even when $k = O(1)$ if $d$ is a growing parameter (as we could conceivably have  $d^k \gg n \cdot 2^k$). An HSG with strong enough parameters would answer the following very natural question.
\begin{question}\label{q:1} Let $k(d):\N \rightarrow \N$ be a non-decreasing function (a slow-growing function, or constant).
If there is an explicit (in the dense representation\footnote{When there are implications to Boolean classes involved, we ought to be careful about the notion of explicitness. For instance, we may have a polynomial $P$ that is explicit in the sense that there is a small algebraic circuit but the circuit involves arbitrary constants. In this instance, we mean that there is deterministic $\poly(d)$-time algorithm  that, on input $1^d$, outputs $P_d$ as a sum of monomials. })
polynomial family $\{P_d\}_{d \in \N}$, where $P_d$ is a $k(d)$-variate polynomial of degree $d$ such that any algebraic circuit computing it has size $d^{\Omega(k)}$, then is PIT in $\mathsf{DTIME}(n^{\poly(k(n))})$ ?

In particular, if $k = O(1)$, does an explicit $k$-variate polynomial family $\set{P_d}_{d\in \N}$, with $\deg(P_d) = d$, that requires size $d^{\Omega(1)}$ imply that PIT is in $\P$?
\end{question}

Note that the hitting set generators by Kabanets and Impagliazzo \cite{KI04} already answer the above question for the specific case when $k = \poly(\log n)$. The above question asks if a \emph{strenghtening} of their hypothesis (to ask for lower bounds for polynomial families on fewer variables) yields a stronger conclusion. In particular, is there a suitable lower bound for algebraic circuits that would yield a complete derandomization of PIT?\\

Another reason for looking beyond the design-based HSGs in the algebraic setting is that by definition, a design-based HSG is combinatorial.
Aesthetically, it seems desirable to have a route from algebraic lower bounds to algebraic pseudorandomness which does not rely on clever combinatorial constructions.

\paragraph{PRGs of Shaltiel \& Umans~\cite{SU05} and Umans~\cite{U03}:}

An alternative to the design-based PRGs in the boolean setting is the generator of Shaltiel and Umans~\cite{SU05}, and a related follow up work of Umans~\cite{U03}.
These generators are quite different from the design-based generators of Nisan and Wigderson~\cite{NW94} and, in particular, appear to be more \emph{algebraic} in their definition and analysis.
We refer the interested reader to the original papers~\cite{SU05, U03} for the formal definitions of these generators and further details.

The algebraic nature of these PRGs makes them good candidates for potential HSGs in the algebraic setting and, indeed, this work was partially motivated by this goal.
However, as far as we understand, it remains unclear whether there is an easy adaptation of these PRGs which works for algebraic circuits.
In particular, the hardness required for the analysis of the PRGs in~\cite{SU05, U03} appears to be inherently functional, i.e.
they assume that it is hard to evaluate the polynomial over some finite field.
In the context of algebraic complexity, the more natural notion of hardness is that it is hard to compute the polynomial syntactically as a formal polynomial via a small algebraic circuit.
\subsection{Our Results}
Our main result is the construction of a hitting-set generator that answers~\autoref{q:1}, for characteristic zero fields. In this paper, we work with slow-growing functions such as $k(d) = o(\sqrt{\log d})$ but we state our results in the introduction for the setting where $k = O(1)$. This is to present our main results without any notational clutter; the general statements can be found in the subsequent sections. 

\begin{theorem}[Main theorem; simpler form of \autoref{thm:derand-from-k-var-lbs}]\label{thm:derand-from-k-var-lbs-informal}
Assume that the underlying field $\F$ has characteristic zero.  Let $k\in \N$ and $\delta > 0$ be arbitrary constants.
Suppose $\set{P_{k,d}}_{d\in \N}$  is an explicit family of $k$-variate polynomials such that  $\deg(P_{k,d}) = d$ and $P_{k,d}$ requires circuits of size at least $d^{\delta}$.
Then, there are explicit hitting sets of size $\poly_{\delta, k}(s)$ for the class $\mathcal{C}(s,s,s)$. 
\end{theorem}

As mentioned earlier, the above theorem can also be extended to slow-growing functions $k:\N\rightarrow \N$ such as $k(d) = o(\sqrt{\log d})$, which would yield non-trivial (albeit superpolynomial sized) hitting sets. The formal statement can be seen in \autoref{thm:derand-from-k-var-lbs}. 
This theorem answers \autoref{q:1} affirmatively. In particular, if $k = O(1)$ and  if we have an explicit $k$-variate hard polynomial family $\set{P_d}_{d\in \N}$ with $\deg P_d = d$ that requires circuits of size $d^{\Omega(1)}$, then PIT $\in \P$. 
As alluded to in the introduction, we do not know of prior constructions of HSGs with these properties.

In addition to being interesting on its own,~\autoref{thm:derand-from-k-var-lbs-informal} leads to the following result which shows that bootstrapping of hitting sets can be done.
We recall the following simpler version of the Polynomial Identity Lemma\footnote{This version below is an immediate consequence of the Combinatorial Nullstellensatz~\cite{Alon99} but a simpler proof is to just use the fact that a univariate degree $d$ polynomial can have at most $d$ roots, one variable at a time (\Cref{lem:Comb-Null}).} \cite{O22,DL78,S80,Z79}. 

\begin{lemma*}[Folklore] Let $f$ be a nonzero $n$-variate polynomial of \emph{individual degree} at most $d$. Then, for any set $S \subseteq \F$ with $|S| > d$, there is a point $\veca \in S^{n}$ such that $f(\veca) \neq 0$. 
\end{lemma*}

This implies that the class of $k$-variate polynomials of individual degree $d$ have hitting sets of size $(d+1)^k$, irrespective of the circuit size. On the other hand, a simple counting argument shows that a random set of size $O(s^2)$ is a hitting set for the class of size $s$ circuits, with high probability.  The following theorem shows that even improving on this bound of $(d+1)^k$ by even \emph{one point} for small circuits would yield to a complete derandomization of PIT.

\begin{restatable}[Bootstrapping hitting sets]{theorem}{bssthm}
  \label{thm:bss}
  Assume that the underlying field $\F$ has characteristic zero. Let $k \in \N$ and  $\delta >0$ be arbitrary constants. Suppose that, for all large enough $s$, there is an explicit hitting set of size \( (s+1)^k - 1 \)
  for the class $\mathcal{C}(k,\text{i-deg: }s, s^\delta)$. Then, there is an explicit hitting set of size $\poly_{\delta, k}(s)$ for $\mathcal{C}(s,s,s)$.   
\end{restatable}

The above theorem answers \autoref{q:2} in a strong sense over characteristic zero fields.
However, it is crucial that we work with the class of algebraic circuits and not subclasses such as algebraic branching programs or algebraic formulas. 

\subsection*{A connection to Shub-Smale's $\tau$-conjecture}
The following conjecture, called the  $\tau$-Conjecture, attributed to Shub and Smale \cite{shub1995,smale98} is a well known conjecture in algebraic complexity.
\begin{definition}[$\tau(f)$]
  For a univariate polynomial $f(x)$, let $\tau(f)$ denote the number of additions, subtractions and multiplications required to compute $f(x)$ starting with with two leaves labelled with $1$ and $x$ respectively (all intermediate constants are expected to be computed from the leaf labelled $1$.

  Let $\operatorname{size}(f)$ refer to its algebraic circuit size, which is the smallest number of gates (with arbitrary constants on the wires) required to compute $f(x)$. 
\end{definition}

\begin{conjecture}[Smale's 4th Problem~\cite{smale98}, or Shub and Smale's $\tau$-conjecture \cite{shub1995}]\label{conj:tau-conj}
  There is an absolute constant $c$ such that the number of distinct integer zeroes for a univariate polynomial $f(x)$ is bounded by $(\tau(f))^c$.
\end{conjecture}

Consider the following strengthened conjecture that uses the algebraic size, $\operatorname{size}(f)$ instead. 

\begin{conjecture}[Strengthened $\tau$-conjecture]\label{conj:tau-conj-mod}
  There is an absolute constant $c$ such that the number of distinct integer zeroes for a univariate polynomial $f(x)$ is bounded by $(\operatorname{size}(f))^c$.
\end{conjecture}

If the above conjecture is true for a constant $c$, then this immediately implies that the polynomial $P_s(x) = (x-1)\cdots (x-s)$, which has $s$ distinct integer roots, requires algebraic circuits of size $s^{1/c}$. In other words, the conjecture immediately implies a lower bound on the size of algebraic circuits computing the polynomial family $\{P_s\}$. Thus, if we invoke this lower bound together with \autoref{thm:derand-from-k-var-lbs-informal}, and using the fact that the coefficients of $P_s$ are not too large, we have the following corollary 
\begin{restatable}{corollary}{tautopit}
\label{cor:tau-to-PIT}
If \autoref{conj:tau-conj-mod} is true then there exist hitting sets, that can be enumerated by a deterministic Turing machine in time $\poly(s)$, for algebraic circuits of size and degree at most $s$. 
\end{restatable}
\noindent
Thus, if the strengthened $\tau$-conjecture \autoref{conj:tau-conj-mod} is true, then PIT is in $\P$.

\medskip

Note that it is important that the coefficients of $P_s$ are bounded by $s! = 2^{O(s \log s)}$. 
Strassen~\cite{Strassen74} showed (amongst other things) an \emph{unconditional} lower bound of $\Omega(d^{1/3})$ on the circuit complexity of the polynomial $Q_d = \sum_{i =0}^d 2^{2^i}x^i$. When invoked with the polynomial $Q_d$, \autoref{thm:derand-from-k-var-lbs-informal} would yield a construction of hitting set consisting of $\poly(s)$ points for general algebraic circuits of size $s$. However, this hitting set is not efficiently enumeratable by a Turing machine in polynomial time since the coefficients of $Q_d$ are too large.  In contrast to this, \autoref{cor:tau-to-PIT} gives a deterministic construction of hitting sets of polynomial size and polynomial bit complexity, although the correctness of the construction is conditional on \autoref{conj:tau-conj-mod}. 


\subsection{An overview of the proof}

All the theorems are a consequence of properties of our hitting-set generator:

\begin{definition}[The generator]\label{def:prg}
For any $k$-variate polynomial $P(\vecz)$, define the map $\g_P: \F^k \times \F^k \to \F^{n+1}$ as follows:
\[
\g_P(\vecz,\vecy) = \left(\D_{0}(P)(\vecz, \vecy), \D_{1}(P)(\vecz, \vecy), \ldots,  \D_{n}(P)(\vecz, \vecy) \right), 
\]
where $\Delta_{i}(P)$ is the homogeneous degree $i$ (in $\vecy$) component in the Taylor expansion of $P(\vecz + \vecy)$, i.e.
\[\D_{i}(P)(\vecz, \vecy) = \sum_{\vece \in \N^k, \abs{\vece} = i} \frac{\vecy^{\vece}}{\vece !} \cdot \frac{\partial P}{\partial \vecz^{\vece}}. \quad\quad \text{(here, $\abs{\vece}:=e_1+\cdots+e_k$ and $\vece! := e_1! \cdots e_k!$)} \qedhere\]
\end{definition}
It is clear that the above definition is $d^{O(k)}$-explicit, where $d = \deg(P)$, as we can express $P$ as a sum of $d^k$ monomials and compute each component of $\g_P(\vecz,\vecy)$ with a small additional cost. The proof proceeds by showing that the above generator $\g_P$ is an HSG provided $P(\vecz)$ is \emph{hard enough}. 

\begin{restatable}[$\g_P$ is an HSG]{theorem}{mainthm}\label{thm:main}
Assume that the underlying field $\F$ has characteristic zero. Let $P$ be a $k$-variate polynomial of degree $d$, for arbitrary integers $k,d > 0$. 
Suppose $P$ cannot\footnote{Note that this statement is meaningful only when $d \gg n^3$ as $P$ always has an algebraic circuit of size at most $d^{k}$.} be computed by algebraic circuits of size  $\tilde{s} = \inparen{s \cdot D \cdot d^3 \cdot n^{10k}}$ for parameters $n,D,s$. Then, for any $(n+1)$-variate polynomial $C(x_0,\ldots, x_n) \in \mathcal{C}(n+1,D,s)$, we have
    \[
      C \neq 0 \Longleftrightarrow  C\circ {\g_P}(\vecz,\vecy) \neq 0.
    \]
\end{restatable}

The proof of \autoref{thm:derand-from-k-var-lbs-informal} is relatively straightforward from \autoref{thm:main} by setting some parameters appropriately and using a \emph{Kronecker-ing trick} (see \autoref{sec:derandomization-from-hard-polynomials}). 
To show that the HSG in \autoref{def:prg} is  indeed a hitting-set generator for $\mathcal{C}(n+1,D,s)$,  we focus our attention on a purported nonzero polynomial $C(\vecx)$, of circuit complexity $s$ and degree $D$, that is not \emph{fooled} by the generator, i.e. $C\circ\g_P$ is identically zero. We use this identity to reconstruct a small circuit for $P$ which contradicts its hardness. This would imply that all polynomials in $\mathcal{C}(n+1,D,s)$ are fooled by the HSG.  

In order to reconstruct a circuit for $P$ from the circuit for $C$, we focus on the so-called \emph{non-degenerate} case and address it in~\autoref{lem:inductive-reconstruction}, which is our key technical lemma. Before discussing the main ideas in the proof of~\autoref{lem:inductive-reconstruction}, we first discuss some of the details of the reduction to the non-degenerate case. 

\paragraph{Reducing to the non-degenerate case : }In the non-degenerate case we insist that, in addition to having $C \circ \g_P = 0$, we have $(\partial_{x_n}C)\circ \g_P \neq 0$; i.e. the derivative of $C$ with respect to the \emph{last} variable $x_n$ is \emph{fooled} by the generator.

Given a nonzero circuit $C$ such that $C \circ \g_P = 0$, and we wish to construct a circuit $C'$ with the above stronger condition.
We may assume without loss of generality that the circuit $C$ is minimal in the sense that all circuits of the same size depending on fewer variables are fooled by the generator; in particular, $C$ depends non-trivially on the last variable $x_n$.

We consider the circuit $\tilde{C}$ obtained by substituting the generator for all coordinates except $x_n$. Interpreting this as a univariate polynomial in $x_n$, with coefficients in $\F(\vecz,\vecy)$, 
\[
  \tilde{C}(x_n) = C(g_1,\ldots, g_{n-1},x_n)
\]
where $\g_P = (g_1,\ldots, g_{n-1},g_n)$. The minimality of $C$ allows us to argue $\tilde{C}(x_n)$ is a nonzero polynomial (for otherwise, $\hat{C} = C(x_1,\ldots, x_{n-1},a) \neq 0$, for a random\footnote{For an infinite field $\F$, by choosing a ``random'' $a\in \F$ we mean choosing $a$ uniformly from a suitably large but finite subset of $\F$.} $a\in \F$, is also not fooled by the generator since $\hat{C} \circ \g_P = \tilde{C}(a) = 0$, contradicting the minimality of $C$). Using the Remainder theorem, this then implies that there must be some $0\leq i \leq \deg_{x_n}(\tilde{C})$ such that
\begin{align*}
  \inparen{\partial_{x_n^i}\tilde{C}}(g_n) & = 0,\\
  \inparen{\partial_{x_n^{i+1}}\tilde{C}}(g_n) & \neq 0.
\end{align*}
Hence, the circuit $C' = \partial_{x_n^i}(C)$ satisfies the \emph{non-degeneracy} condition. Standard interpolation arguments shows that the circuit complexity of $C'$ is at most $O(sD)$ and we work with this circuit~\footnote{This is also where the dependency of the hardness of $P$ on the degree $D$ appears. We suspect that this dependency can be removed.}. 

\paragraph*{The proof of~\autoref{lem:inductive-reconstruction} : }
The proof of the lemma can be viewed as a variant of the standard Newton Iteration (or Hensel lifting) based argument often used in the context of root finding, although there are some crucial differences.
We iteratively construct the polynomial $P(\vecz)$ one homogeneous component at a time (recall that $P(\vecz)$ is a $k$-variate polynomial of degree $d$).
In fact, our induction hypothesis needs to be a bit stronger than this.
For our proof, we maintain the invariant that at the end of the $i^{th}$ iteration, we have a multi-output circuit which computes all the partial derivatives of order at most $n$ of all the homogeneous components of $P(\vecz)$ of degree at most $i + n$.
However, for this overview, we ignore this technicality and pretend that we are directly working with the homogeneous components of $P(\vecz)$.

For the base case, we assume that we have access to all the homogeneous components of $P$ of degree at most $n$, which are homogeneous polynomials of degree at most $n$ on $k$ variables and are trivially computable by a circuit of size at most $n^{5k}$.
the presumed hardness of $P$ for $d \gg n$.
Thus, we have $n$ homogeneous components of $P(\vecz)$, and the goal is to use them and the non-degeneracy assumption to reconstruct all of $P$.
Let us assume that we have already computed $P_0,\ldots, P_{i}$, where $P_j$ is the homogeneous component of $P(\vecz)$ of degree equal to $j$.
We now focus on recovering the homogeneous component $P_{i+1}$ of degree equal to $i + 1$.
Observe that $\Delta_n(P_{i+1})$ is a homogeneous (in $\vecz$) polynomial of degree $(i-n+1)$.
We show that given the non-degeneracy condition in the hypothesis of the lemma, there is a small circuit such that, modulo the ideal ${\inangle{\vecz}^{i-n+2}}$, it computes $\D_n(P_{i + 1})(\vecz,\vecy)$.
Since $\D_n(P_{i+1})(\vecz,\vecy)$ is essentially a \emph{generic} linear combination of $n$-th order derivatives of $P_{i+1}(\vecz)$, it is not hard to show\footnote{In particular, this part of the argument relies on the stronger hypothesis that we have access to each of the order $n$ partial derivatives of $P_0, P_1, \ldots, P_{i}$.}
that we can obtain a small circuit that outputs each of the $n$-th order partial derivatives of $P_{i+1}(\vecz)$, modulo higher degree monomials.
Then we would be able to reconstruct $P_{i+1}(\vecz) \bmod{\inangle{\vecz}^{i+2}}$ via repeated applications of the Euler's differentiation formula for homogeneous polynomials:

\begin{fact}[Euler's formula for differentiation of homogeneous polynomials]
  \label{fact:euler}
If $A(x_1,\ldots, x_k)$ is a homogeneous polynomial of degree $t$, then $\sum_{i=1}^k x_i \cdot \partial_{x_i} A = t \cdot A(x_1,\ldots, x_k)$. 
\end{fact}
One crucial point in this entire reconstruction is that each step of the reconstruction only incurs an \emph{additive} blow-up in size and hence can be repeated for polynomially many steps to recover each homogeneous part of $P$ (\autoref{fig:Blahprime} in \autoref{sec:main-theorem} contains a pictorial description of the inductive step). 

However, we are still left with the task of getting rid of the higher order terms as we have only constructed a circuit computing $P_{i+1}$ modulo the ideal $\inangle{\vecz}^{i+2}$.
We need to extract the lowest degree homogeneous parts from the outputs.
The standard way to proceed here is to perform a \emph{homogenisation} but this needs to be done carefully.
Typically, extracting a certain homogeneous part incurs a multiplicative blow-up in size which is unaffordable in this setting as this needs to be performed for $d$ steps!
Fortunately, the structure of the circuit built has the property that each inductive step adds more gates atop the outputs of the previous steps and hence it suffices to \emph{only homogenise} the newly added gates.
This allows us to argue that each inductive step incurs only an additive blow-up of $d^2 \cdot n^{O(k)} \cdot sD$. By performing the reconstruction step to extract all homogeneous components of $P$, we can construct a small enough circuit for $P$, contradicting the hardness of $P$. That would complete the proof of \autoref{thm:main}. 


\paragraph{Similarities with PRGs of Shaltiel and Umans~\cite{SU05, U03} and list decoding algorithms of Kopparty~\cite{K15}: } We remark that at a high level, our construction of the HSG was inspired by the constructions by Shaltiel and Umans~\cite{SU05, U03}, although the precise form of our generator seems different from those in~\cite{SU05, U03}.  We also note that the set up of induction we have in the proof of~\autoref{lem:inductive-reconstruction} is very similar to the set up used by Kopparty~\cite{K15} in the context of list decoding  Multiplicity codes. More precisely, our induction is similar to what is used in constructing a power series expansion of  a non-degenerate solution of the univariate Cauchy-Kovalevski differential equations, which are used in~\cite{K15}. The key difference is that although we work with a multivariate setting (and hence deal with a \emph{partial} differential equation of high order and high degree), the iterative proof of~\autoref{lem:inductive-reconstruction} resembles the list decoding algorithm for univariate multiplicity codes in~\cite{K15} (which deals with an \emph{ordinary} differential equation of high order and high degree). It appears to be of interest  to understand this analogy further. 

\paragraph{Relating \autoref{thm:bss} and results of Jansen and
  Santhanam~\cite{JS12}: }The hypothesis of \autoref{thm:bss}, for $k = 1$, bears some similarities with a result of Jansen and Santhanam~\cite{JS12}. Jansen and Santhanam showed, roughly speaking\footnote{The exact theorem statement is more complicated, and also handles a larger range of parameters, but the above description contains the essence of their main statement.}, that if the class $\mathcal{C}(1,d,\polylog(d))$  has a hitting set of size $d$ (one point fewer than the trivial hitting set of size $d+1$) that \emph{can be enumerated by small, succinctly encodable $\mathsf{TC}^0$ circuits}, then $\mathrm{Perm}$ does not have polynomial sized \emph{constant-free} algebraic circuits.

Our hypothesis, for $k=1$, is similar in the sense that it requires a saving of one point from the trivial hitting set for the appropriate class, but we only need the standard notion of explicitness for the hitting set. However, the two results work for different ranges of parameters --- their hypothesis requires the hitting set for $C(1,d,\polylog(d))$ whereas ours requires the same for $C(1,d,d^\delta)$ for some constant $\delta > 0$. Furthermore, our conclusion is also in terms of derandomizing PIT, and not a lower bound. It would be interesting to see if our hypothesis is also sufficient to obtain a similar conclusion as Jansen and Santhanam.

\section{Notation and preliminaries}

We follow the following notations. 
\begin{itemize}
\item We use $\N$ to denote the set of natural numbers $\{1, 2, 3, \dots\}$.

\item Throughout the paper, we think of $\F$ as a field of characteristic zero (or large enough). 
  
\item We use boldface letters such as $\vecz$ to denote tuples: $\vecz = \inparen{z_1, z_2, \ldots, z_k}$; in almost all instances, the length of the tuple will be clear from context. For an exponent vector $\vece$, we shall use $\vecz^{\vece}$ to denote the monomial $z_1^{e_1}\cdots z_k^{e_k}$. Let $\abs{\vece} := \sum e_i$. 

\item We use $\partial_{\vecz^{\vece}}(P(\vecz))$ to denote the partial derivative $\frac{\partial^{\abs{\vece}}(P)}{\partial \vecz^{\vece}}$. 
  
\item We use $\inangle{\vecz}^{i}$ to denote the ideal in $\F[\vecz]$ generated by all degree $i$ monomials in $\vecz$.

\item We use $\mathcal{P}(k,d)$ to denote the class of $k$-variate polynomials of degree at most $d$. 
\end{itemize}

\subsection{Algebraic circuits}
\begin{definition}[Algebraic circuits]\label{defn:alg-circuits}
  An \emph{algebraic circuit} is a directed acyclic graph whose internal nodes are labelled by $+$ and $\times$, with leaves (in-degree zero nodes) labelled by variables or field constants. The addition and multiplication gates naturally compute the sum and product of its children respectively. Edges may carry arbitrary field constants whose semantic behaviour is to scale the value passing through it.

  The \emph{size} of an algebraic circuit refers to the number of gates in the circuit (and is independent of the constants that may appear on the edges/wires).
\end{definition}

We use the notation $\mathcal{C}(n, d, s)$ to denote the class of $n$-variate polynomials of degree bounded by $d$ that are computable by circuits of size at most $s$. Similarly, $\mathcal{C}(n, \text{i-deg: }d,s)$ refers to the class of $n$-variate polynomials of \emph{individual degree} bounded by $d$  computable by circuits of size $s$.

\subsection{PIT preliminaries}

The following well-known lemma gives an exponential (in the number of variables) sized hitting set for low degree polynomials. 

\begin{lemmawp}[\cite{O22,DL78,S80,Z79}]\label{lem:Schwartz-Zippel} Let $f$ be a nonzero $n$-variate polynomial of degree at most $d$. Then for any set $S \subseteq \F$ with $|S|> d$, there is a point $\veca \in S^{n}$ such that $f(\veca) \neq 0$. 
\end{lemmawp}

\noindent
The following is a variant for the case of polynomials of \emph{individual degree} bounded by $d$. 

\begin{lemma}\label{lem:Comb-Null} Let $f$ be a nonzero $n$-variate polynomial of \emph{individual degree} at most $d$. Then, for any set $S \subseteq \F$ with $|S| > d$, there is a point $\veca \in S^{n}$ such that $f(\veca) \neq 0$. 
\end{lemma}
\begin{proof}
  Let $S\subseteq \F$ with $|S| \geq d+1$. We may interpret $f$ as a univariate polynomial in $x_1$ of degree at most $d$ with coefficients from $\F\inparen{x_2,\ldots, x_n}$. Since a nonzero univariate polynomial of degree $d$ can have at most $d$ roots, there must be some $a_1 \in S$ such that $f(a_1,x_2,\ldots, x_n) \neq 0$. Repeating this argument for variables $x_2,\ldots, x_n$, the lemma follows. 
\end{proof}

\noindent
A more general statement is the following Combinatorial Nullstellensatz due to Alon~\cite{Alon99}.

\begin{theorem}[Combinatorial Nullstellensatz \cite{Alon99}]\label{thm:CNS}
Let $f(x_1,\ldots, x_n)$ be a nonzero polynomial with $\deg(f) \leq t_1 + \cdots + t_n$, and suppose the coefficient of $x_1^{t_1}\cdots x_n^{t_n}$ is nonzero in $f$. If $S_1,\ldots, S_n \subseteq \F$ with $|S_i| \geq (t_i+1)$, then there is a point $\veca = (a_1,\ldots, a_n) \in S_1\times \cdots \times S_n$ such that $f(\veca) \neq 0$. 
\end{theorem}

It is also known that existence of non-trivial hitting sets for a class $\mathcal{C}$ can be used to construct hard polynomials. We state the version that corresponds to the individual degree as opposed to total degree. 

\begin{theorem}[Heintz and Schnorr~\cite{HS80},
  Agrawal~\cite{A05a}]\label{thm:HS} Let $H(n, \text{i-deg: }d, s)$ be an explicit hitting set for the class $\mathcal{C}(n,\text{i-deg: }d, s)$.
Then, for every $k\leq n$ and $d'$ such that $d'\leq d$ and $(d'+1)^{k} > \abs{H(n, \text{i-deg: }d, s)}$, there is a nonzero polynomial on $k$ variables and individual degree $d'$ that vanishes on the hitting set $H(n,\text{i-deg: }d, s)$, and hence cannot be computed by a circuit of size $s$.
\end{theorem}

Finally, we need the following notion of \emph{interpolating sets} for a class of polynomials. 

\begin{definition}[Interpolating sets for $\mathcal{P}(k,d)$] Let $M_{k,d}$ denote the number of $k$-variate monomials of degree at most $d$. That is,
  $M_{k,d} = \binom{k+d}{d}$. 

  A set of points $\veca_1,\ldots, \veca_r \in \F^{k}$ is said to be an \emph{interpolating set for $\mathcal{P}(k,d)$} if the vectors
  \[
    \setdef{\inparen{\veca_i^{\vece}\;:\; \vece \in \Z_{\geq 0}^k\;,\; \abs{\vece}\leq d}}{i\in [r]} \subset \F^{M_{k,d}}
  \]
  form a spanning set for $\F^{M_{k,d}}$.  

  \noindent
  Equivalently, for every $\vece \in \Z_{\geq 0}^k$ with $\abs{\vece}\leq d$, there exists field constants $\beta_1,\ldots, \beta_{r}$ such that for every $f(\vecz) \in \mathcal{P}(k,d)$, we have
  \[
    \coeff_{\vecz^{\vece}}(f) = \sum_{i=1}^r \beta_i \cdot f(\veca_i).\qedhere
  \]
\end{definition}

A canonical example of an interpolating set for $\mathcal{P}(k,d)$ is $S^k = \setdef{(s_1,\ldots, s_k)}{s_i \in S\;\forall i}$ where  $S \subseteq \F$ is a set of at least $(d+1)$ distinct field elements. The following well-known proposition says that we can find interpolating sets inside a suitably large grid, even in the presence of additional polynomial constraints. 

\begin{proposition}[Interpolating sets with additional polynomial constraints]\label{prop:interpolating}
  Let $d,k$ be arbitrary positive integers, and any polynomial $Q(\vecx)$ of degree $D$. If $S \subseteq \F$ with $|S| \geq (D+d+1)$ there is a set of points $\veca_1,\ldots, \veca_{M_{k,d}} \in S^k$ such that $Q(\veca_i) \neq 0$ for all $i \in [M_{k,d}]$ and $\set{\veca_1,\ldots, \veca_{M_{k,d}}}$ is an interpolating set for $\mathcal{P}(k,d)$. 
\end{proposition}
\begin{proof}
Consider the following $M_{k,d} \times M_{k,d}$ matrix $\Lambda$: 
  \begin{quote}
    The rows of $\Lambda$ are indexed by $[M_{k,d}]$ and the columns indexed by monomials $\vecx^{\vece}$ of degree at most $d$, and the entry at $(i,\vecx^{\vece})$ is $\veca_i^{\vece}$ (the monomial $\vecx^\vece$ evaluated at $\veca_i$). 
  \end{quote}
  The condition that $\set{\veca_1,\ldots, \veca_{M_{k,d}}}$ is an interpolating set can be characterised by the nonzeroness of the determinant of $\Lambda$. Hence, the condition that $\veca_1,\ldots, \veca_{M_{k,d}}$ forms an interpolating set and $Q(\veca_i) \neq 0$ for all $i$ can be charcterised by the nonzeroness of the polynomial
  \[
   \Gamma(\veca_1,\ldots, \veca_{M_{k,d}}) := \Lambda(\veca_1,\ldots, \veca_{M_{k,d}}) \cdot \prod_{i = 1}^{M_{k,d}} Q(\veca_i). 
  \]
    Let $\setdef{a_{i,j}}{i\in [M_{k,d}] \;,\; j \in [k]}$ be indeterminates, corresponding to the coordinates of the points  $\veca_1,\ldots, \veca_{M_{k,d}}$. Although $\Gamma$ is a $k\cdot M_{k,d}$-variate polynomial of total degree $O((d+D)M_{k,d})$,  the individual degree is at most $(D + d)$ (since each row of $\Gamma$ is indexed by a disjoint set of variables). Thus, by \autoref{thm:CNS}, there are points $\veca_1,\ldots, \veca_{M_{k,d}} \in [D+d+1]^k$ that forms an interpolating set for $\mathcal{P}(k,d)$ and also keeps $Q$ nonzero. 
\end{proof}

\subsection{Homogenisation} \label{sec:homogenisation}

\begin{definition}[Homogeneous circuits] \label{defn:hom-circuits}
  A circuit $C$ is said to be \emph{homogeneous} if every gate of the circuit computes a homogeneous polynomial.

  For a non-homogeneous polynomial $f$, we shall say that a \emph{homogeneous (multi-output) circuit} $C$ computes $f$ if the outputs of $C$ are the homogeneous parts of $f$. 
\end{definition}

\begin{lemmawp}[Strassen's homogenisation]
  Let $C$ be a circuit of size $s$ computing a homogeneous polynomial of degree $d$. Then, there is a \emph{homogeneous circuit} $C'$ of size at most $O(sd^2)$ computing the same polynomial.
\end{lemmawp}

\begin{lemma}[Partial homogenisation]\label{lem:partial-hom}
  Let $C$ be a multi-output homogeneous circuit of size $s$, with $m$ outputs computing homogeneous polynomials $f_1,\ldots, f_m$. Suppose $C'$ is a multi-output, $m$-input circuit of size $s'$. Then, there is a \emph{homogeneous circuit} $D$ of size at most $s + O(s' \cdot d^2)$ computing $C' \circ C = C'(f_1,\ldots, f_m)$, where $d$ is the maximum degree of the outputs of $C'(f_1,\ldots, f_m)$. 
\end{lemma}
\begin{proof}
  We follow the standard homogenisation procedure, but applied only to the circuit $C'$ --- we replace each gate $g\in C'$ by copies $g_0,\ldots, g_d$ and add the following connections:
  \[
    g_a = \begin{cases}
      h^{(L)}_a + h^{(R)}_a & \text{if }g = h^{(L)} + h^{(R)}\\
      \sum_{b=0}^ah^{(L)}_b\times h^{(R)}_{a-b} & \text{if }g = h^{(L)} \times h^{(R)}
    \end{cases}
  \]
  and any leaf $\ell$ of $C'$ labeled with the $i$-th variable is renamed as $\ell_{\deg(f_i)}$. Since each of the $f_i$ is a homogeneous polynomial, it is immediate to see that the above transformation yields a homogeneous circuit for $C' \circ C$ of size at most $s + O(s' d^2)$. 
\end{proof}

\subsection{The Generator}

For a $k$-variate polynomial $P$, let $\Del_i(P)(\vecz,\vecy) \in \F[\vecz,\vecy]$ be defined as
\[
  \Del_i(P) = \sum_{\vece: \abs{\vece} = n} \pfrac{\vecy^\vece}{\vece!}\cdot \partial_{\vecz^{\vece}}(P)
\]
where $\vece! = e_1! \cdots e_k!$. The generator with respect to $P$ is defined as follows:
\[
  \g_P(\vecz,\vecy) = \inparen{\Del_0(P)(\vecz,\vecy),\ldots, \Del_n(P)(\vecz,\vecy)}. 
\]
The following are some simple observations about the operator $\Del$. 

\begin{observation}\label{obs:Del-under-shift}
  For any $\veca \in \F^k$, if $P'(\vecz) = P(\vecz + \veca)$, then for all $i$ we have \(\Del_i(P')(\vecz,\vecy) = \Del_i(P)(\vecz + \veca,\vecy).\)
  Hence, $\g_{P'}(\vecz,\vecy) = \g_P(\vecz + \veca,\vecy)$. 
\end{observation}

\begin{observation}\label{obs:Del-under-modzi}
  Let $P(\vecz)$ and $Q(\vecz)$ be polynomials such that $P = Q \bmod{\inangle{\vecz}^{j}}$. Then, for any $i \leq j$, we have $\Del_i(P) = \Del_i(Q) \bmod{\inangle{\vecz}^{j-i}}$. 
\end{observation}

\section{The Main Theorem}
\label{sec:main-theorem}

We start by recalling the theorem that states that our generator is indeed a hitting-set generator. 

\mainthm*

\noindent
The rest of this section would be devoted to the proof of this theorem. \\

Let us assume the contrary. That is, there is a circuit $C(\vecx)$ of size $s$ and degree $D$ such that $C \neq 0$ but $C \circ \g_P(\vecz,\vecy) = 0$. We shall assume, without loss of generality, that $C$ depends non-trivially on the variable $x_n$ and that no circuit $C'(\vecx)$ of size $s$ and degree $D$  with $C'$ depending on fewer variables satisfy $C' \neq 0$ and $C' \circ \g_P(\vecz,\vecy) = 0$.

The proof will proceed by inductively building a circuit that computes each homogeneous part of $P$ but we would need the following preprocessing step.

\paragraph*{Preprocessing the circuit:}

Let $C(x_0,\ldots, x_n)$ be the minimal (in terms of number of variables) size $s$ circuit that is \emph{not} fooled by $\g_P(\vecz,\vecy)$.
That is, $C \circ \g_P(\vecz,\vecy) = 0$.

\begin{claim}
There is some $i\geq 0$ such that
\begin{align*}
  \partial_{x_n^i}(C) \circ \g_P(\vecz,\vecy) &= 0,\\
  \partial_{x_n^{i+1}}(C) \circ \g_P(\vecz,\vecy) &\neq 0.
\end{align*}
\end{claim}
\begin{proof}
  Let $\g_P(\vecz,\vecy) = (g_0,\ldots, g_n)$. Consider the polynomial $C(g_0,\ldots, g_{n-1}, x_n)$ (which is the application of the generator $\g_P$ to all coordinates except one) as an element $\hat{C}(x_n) \in \F(\vecy, \vecz)[x_n]$.

  {\bf Case 1: $(\hat{C} = 0)$}
  
  In this case, using \autoref{lem:Comb-Null}, let $a \in \set{0,1,2,\ldots, D}$ such that $C(x_0,\ldots, x_{n-1},a) \neq 0$. Then, $C(x_0,\ldots, x_{n-1},a)$ is also a nonzero polynomial computable by a size $s$ circuit,  depending on fewer than $n+1$ variables, that also vanishes on $\g_P(\vecz,\vecy)$. This contradicts the minimality of $C$.
   \medskip

  {\bf Case 2: $(\hat{C} \neq 0)$}
  
  Let $r = \deg_{x_n}(\hat{C})$. Note that if $(\partial_{x_n^i}(\hat{C}))(g_n) = 0$ for all $0\leq i \leq t$, then $(x_n - g_n)^{t+1}$ divides $\hat{C}$. Since $\deg_{x_n}(\hat{C})$ is $r$, the largest $e$ such that $(x_n - g_n)^e$ divides $\hat{C}$ is at most $r$. Hence, there must be some $t \leq r-1$ such that $\partial_{x_n^{t}}(\hat{C})(g_n) \neq 0$. Since $\hat{C}(g_n) = C \circ \g_P(\vecz, \vecy) = 0$ and $\partial_{x_n^{t}}(\hat{C})(g_n) \neq 0$, there must be an intermediate derivative where  a switch from zero to nonzero occurs. Hence, there must be some $i < r$ such that
\begin{align*}
  \partial_{x_n^i}(C) \circ \g_P(\vecz,\vecy) &= \partial_{x_n^{i}}(\hat{C})(g_n) = 0,\\
  \partial_{x_n^{i+1}}(C) \circ \g_P(\vecz,\vecy) &= \partial_{x_n^{i+1}}(\hat{C})(g_n) \neq 0.\qedhere
\end{align*}
\end{proof}

Let $C' = \partial_{x_n^i}(C)$.  In what follows, we will work with $C'$ instead of $C$. By interpolation, its size $s' \leq s \cdot D$ (where, recall, $D \geq \deg(C)$) and we now have
\begin{align*}
  C' \circ \g_P(\vecz,\vecy) &= 0,\\
  \partial_{x_n}(C') \circ \g_P(\vecz,\vecy) &\neq 0.
\end{align*}

\noindent
Without loss of generality (by translating $\vecz$ by a point in $\set{0,1,\ldots, dD}^k$ if necessary, via \autoref{obs:Del-under-shift}), we may assume that
\[
  \inparen{\partial_{x_n}(C') \circ \g_P(\mathbf{0},\vecy)}\;=:\; \Psi(\vecy) \neq 0.\]
Let $P = P_0 + P_1 + \cdots + P_d$ be the decomposition into homogeneous parts, with $P_i$ being the homogeneous part of degree $i$, and let $P_{\leq r} := \sum_{i \leq r} P_i$.\\

\paragraph*{The reconstruction: } We now proceed to describe the set up of the inductive reconstruction of a \emph{small} circuit for $P$. The induction is on a parameter $j$ which takes values from $0$ up to $d-n$. At the end of the $j^{th}$ step, we would have a circuit which computes all partial derivatives of order at most $n$ of homogeneous components of $P$ of degree at most $j + n$. We now describe the steps of the induction argument more formally. 

{\bf Base case $(j=0)$ :} Each $\partial_{\vecz^{\vece}} P_{\ell}$ for $\abs{\vece}\leq n$ and $\ell \leq n$ can be explicitly written as a sum of at most $N := \binom{n+k}{k}$ monomials. Hence, there is a circuit $B_{0}$ of size $s_0 = N^{2}$ that computes $\setdef{\partial_{\vecz^{\vece}}(P_\ell)}{0\leq \ell \leq n\;,\; \abs{\vece}\leq n}$. 
\medskip

{\bf Induction hypothesis:} There is a circuit $B_{j-1}(\vecz)$ of size at most $s_{j-1}$, with $N(n+j-1)$ outputs that computes $\partial_{\vecz^\vece} P_{\ell}$ for each $\vece$ with $\abs{\vece}\leq n$ and $\ell \leq n+j-1$.

\medskip

{\bf Induction step:} To construct a circuit $B_j(\vecz)$ of size at most $s_j$ (to be defined shortly) that computes $\partial_{\vecz^\vece} P_{\ell}$ for each $\vece$ with $\abs{\vece}\leq n$ and $\ell \leq n+j$.\\

Recall $N = \binom{n+k}{n}$, the number of $k$-variate, degree $n$ monomials. Recall that $\Psi(\vecy) = (\partial_{x_n}(C') \circ \g_P(\mathbf{0},\vecy))$.  We shall say that $\veca \in \F^k$ is ``good'' if $\Psi(\veca) \neq 0$. By \autoref{prop:interpolating}, there exists $\set{\veca_1,\ldots, \veca_N} \subset \set{0,\ldots, \poly(d,D,n)}^k$ is a set of ``good'' points and also an interpolating set for $\mathcal{P}(k,n)$. Let $\Gamma_{j-1,\veca}$ be defined as
\[
  \Gamma_{j-1,\veca} := (\Del_0(P_{\leq n+j-1})(\vecz,\veca),\ldots, \Del_n(P_{\leq n+j-1})(\vecz,\veca)).
\]

\begin{lemma}\label{lem:inductive-reconstruction}
  Let $\veca \in \F^k$ be such that $0 \neq \Psi(\veca) = (\partial_{x_n} C')\circ \g_P(\mathbf{0},\veca)$. Then,
  \[
    \pfrac{-1}{\Psi(\veca)} \cdot C'(\Gamma_{j-1,\veca})  = \Del_n(P_{n+j})(\vecz,\veca) \bmod{\inangle{\vecz}^{j+1}}.
  \]
\end{lemma}

\noindent
We will defer the proof of this lemma to the end of the section and finish the rest of the proof.\\

We can begin with the circuit $B_{j-1}(\vecz)$ that computes every $\partial_{\vecz^{\vece}}(P_{\ell})$ for $\abs{\vece} \leq n$ and $\ell \leq n+j-1$.
By taking suitable linear combinations of the output gates, we can create a new circuit $B$, of size at most $s_{j-1} + N^{5}$, that computes $\setdef{\Gamma_{j-1,\veca_t}}{t\in [N]}$.
Using \autoref{lem:inductive-reconstruction} for each $\veca_i$, we then obtain a circuit of size $s_{j-1} + N^{5} + s'\cdot N$ that computes $\setdef{\Del_n(P_{n+j})(\vecz,\veca_t)}{t\in [N]}$ modulo the ideal $ {\inangle{\vecz}^{j+1}}$.

By definition, $\Del_n(P_{n+j})(\vecz,\veca)$ is a suitable linear combination of the $n$-th order partial derivatives of  $P_{n+j}(\vecz)$.
 As $\set{\veca_1,\ldots, \veca_N}$ was chosen to be an interpolating set, each $\partial_{\vecz^{\vece}}(P_{n+j})$ with $\abs{\vece} = n$ can be written as a suitable linear combination of $\setdef{\Del_n(P_{n+j})(\vecz,\veca_t)}{t\in [N]}$. 
Furthermore, since $P_{n+j}$ is a homogeneous polynomial, we can also compute all its lower order derivatives via repeated applications of Euler's formula (\autoref{fact:euler}). 
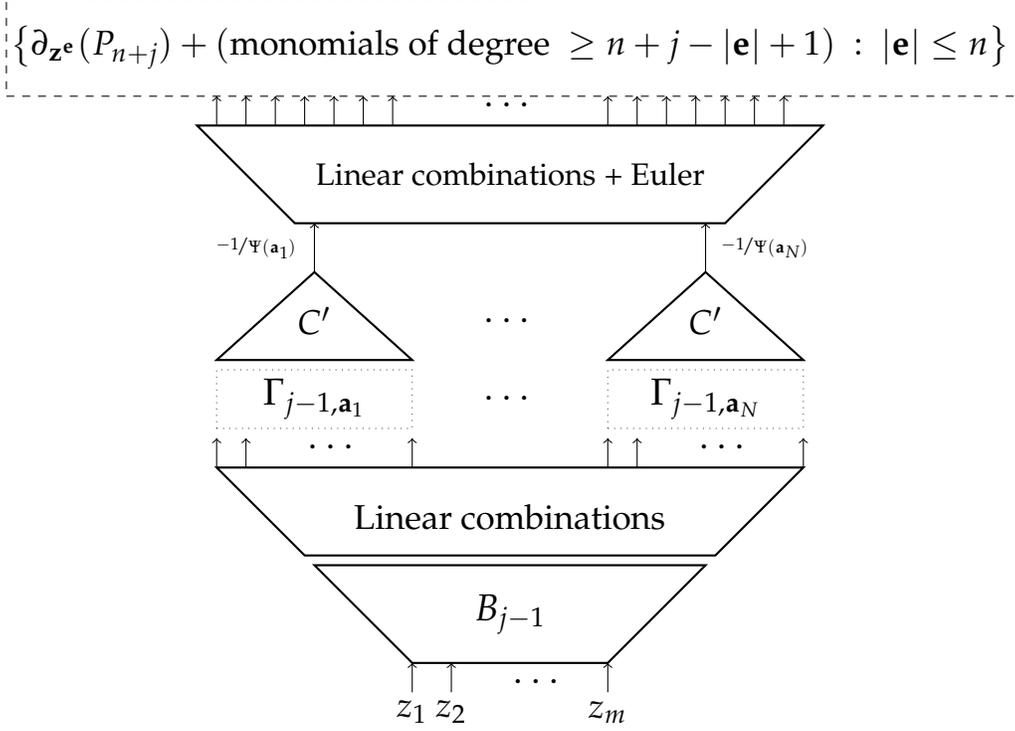
\begin{figure}
  \begin{center}
  \begin{tikzpicture}[transform shape, scale=1.3]
    \draw[<-] (-1,-1 ) -- +(0,-0.3);
    \node at (-1,-1.5) {{\small $z_1$}};
    \draw[<-] (-0.6,-1 ) -- +(0,-0.3);
    \node at (-0.6,-1.5) {{\small $z_2$}};
    \draw[<-] (1,-1 ) -- +(0,-0.3);
    \node at (1,-1.5) {{\small $z_m$}};
    \node at (0.3,-1.2) {$\cdots$};

    \draw[thick] (-1,-1) -- ++(-1,1) -- ++(4,0) -- ++(-1,-1) -- cycle;
    \node at (0,-0.5) {$B_{j-1}$};
    \draw[thick] (-2.1,0.1) -- ++(-0.9,0.9) -- ++(6,0) -- ++(-0.9,-0.9) -- cycle;
    \node at (0, 0.5) {\small Linear combinations};

    \draw[->] (-3,1 ) -- +(0,0.3);
    \draw[->] (-2.7,1 ) -- +(0,0.3);
    \draw[->] (-1,1 ) -- +(0,0.3);
    \node at (-1.8,1.2) {$\cdots$};
    \draw[dotted,black!70] (-3,1.4) rectangle (-1,2);
    \node at (-2,1.7) {$\Gamma_{j-1,\veca_1}$};

    \draw[->] (1,1 ) -- +(0,0.3);
    \draw[->] (1.3,1 ) -- +(0,0.3);
    \draw[->] (3,1 ) -- +(0,0.3);
    \node at (2.2,1.2) {$\cdots$};
    \draw[dotted,black!70] (1,1.4) rectangle (3,2);
    \node at (2,1.7) {$\Gamma_{j-1,\veca_N}$};
    \node at (0,1.7) {$\cdots$};

    \draw[thick] (-3,2.1) -- ++(2,0) -- ++(-1,0.9) -- cycle;
    \node at (-2,2.5) {\small $C'$};
    \draw[->] (-2,3) -- ++(0,0.5);
    \node at (-2.6,3.25) {\tiny $\nicefrac{-1}{\Psi(\veca_1)}$};
    
    \draw[thick] (1,2.1) -- ++(2,0) -- ++(-1,0.9) -- cycle;
    \node at (2,2.5) {\small $C'$};
    \draw[->] (2,3) -- ++(0,0.5);
    \node at (2.6,3.25) {\tiny $\nicefrac{-1}{\Psi(\veca_N)}$};

    \node at (0,2.5) {$\cdots$};
    \draw[thick] (-2.2,3.5) -- ++(-1,1) -- ++(6.4,0) -- ++(-1,-1) -- cycle;
    \node at (0,4) {\footnotesize Linear combinations + Euler};

    \foreach\x in {-3,-2.7,-2.4,...,-1} \draw[->] (\x,4.5) -- ++(0,0.3);
    \foreach\x in {1,1.3,1.6,...,3} \draw[->] (\x,4.5) -- ++(0,0.3);
    \node at (0,4.7) {$\cdots$};
    \draw[dashed] (-5.15,4.8) rectangle (5.2,5.8);
    \node at (0,5.3) {\small $\setdef{\partial_{\vecz^{\vece}}(P_{n+j}) + (\text{monomials of degree } \geq n+j-\abs{\vece} + 1)}{\abs{\vece} \leq n}$};
    
  \end{tikzpicture}
\end{center}
\caption{Pictorial representation of $B_{j}'$}
\label{fig:Blahprime}
\end{figure}
Overall, combined with the outputs of $B_{j-1}(\vecz)$, we have a circuit $B_j'(\vecz)$ (shown in \autoref{fig:Blahprime}) of size $s_{j-1} + N^{10} + s ' N$ that computes
\[
  \setdef{\partial_{\vecz^\vece}(P_\ell)}{\abs{\vece} \leq n\;,\;\ell \leq n+j-1}\;\union\; \setdef{\partial_{\vecz^\vece}(\tilde{P}_{n+j})  }{\abs{\vece} \leq n}, 
  \]
  where $\partial_{\vecz^\vece}(\tilde{P}_{n+j}) \mod \inangle{\vecz}^{n+ j - \abs{\vece}+1} \equiv \partial_{\vecz^\vece}({P}_{n+j})$  for every $\abs{\vece} \leq n$. At this point, the only task left to do is extracting the lowest degree homogeneous components of these outputs. 
  
  From \autoref{fig:Blahprime}, the circuit $B_j'$ is a composition of a circuit of size $N^{10} + s'N$ over the homogeneous circuit $B_{j-1}$ of size $s_{j-1}$. By \autoref{lem:partial-hom}, we extract the lowest degree homogeneous parts of the outputs of $B_j'$ by constructing an equivalent \emph{homogeneous} circuit of size at most $s_{j-1} + ((N^{10} + s'N) \cdot d^2)$ that computes
  \[
    \setdef{\partial_{\vecz^\vece}(P_\ell)}{\abs{\vece} \leq n\;,\;\ell \leq n+j}.
  \]
  This completes the induction step.

  \medskip
  
  Unraveling the induction for $d-n$ steps, we eventually obtain a circuit of size at most $s_{d-n} = O(s' \cdot d^3 \cdot N^{10}) = O(s \cdot D \cdot d^3 \cdot n^{O(k)})$ that computes all the partial derivatives of order at most $n$ of $P_0,\ldots, P_{d}$, which in particular includes $P_0,\ldots, P_d$ that can be summed to produce a circuit for $P$ of size $O(s \cdot D \cdot d^3 \cdot n^{O(k)})$. However, this contradicts the hardness assumption of $P$. Hence, it must be the case that $C \circ \g_P(\vecz,\vecy) \neq 0$. This completes the proof of the main theorem barring the proof of \autoref{lem:inductive-reconstruction}; we address this next.  \qed \raisebox{0.2 em}{\scriptsize (\autoref{thm:main})}

\subsection{Proof of Lemma~\ref{lem:inductive-reconstruction}}

We are given $\Gamma_{j-1,\veca} = (\Del_0(P_{\leq n+j-1})(\vecz,\veca),\ldots, \Del_n(P_{\leq n+j-1})(\vecz,\veca))$. From the assumption on $C'$, we have
\begin{align*}
  0 &=  C'(\Del_0(P)(\vecz,\veca),\ldots, \Del_n(P)(\vecz,\veca))\\
  \implies 0  &= C'(\Del_0(P)(\vecz,\veca),\ldots, \Del_n(P)(\vecz,\veca)) \bmod{\inangle{\vecz}^{j+1}}.
\end{align*}

By \autoref{obs:Del-under-modzi}, we have that $\Del_i(P)(\vecz,\veca) = \Del_i(P_{\leq n+j-1})(\vecz,\veca) \bmod{\inangle{\vecz}^{j+1}}$ for all $i \leq  n-1$, and $\Del_n(P)(\vecz,\veca) = \Del_n(P_{\leq n+j-1})(\vecz,\veca) + \Del_n(P_{n+j})(\vecz,\veca) \bmod{\inangle{\vecz}^{j+1}}$. For the sake of brevity, let $R_i = \Del_i(P_{\leq n + j -1})(\vecz,\veca)$ for $0\leq i\leq n$ and $A = \Del_n(P_{n+j})(\vecz,\veca)$. Therefore,
\begin{align*}
0& = C'(R_0, R_1,\ldots, R_{n-1}, R_n + A)\bmod{\inangle{\vecz}}^{j+1}.
\end{align*}
\noindent We now expand the above expression as a \emph{univariate in $A$} (or in other words, perform a Taylor expansion of the polynomial $C'$ around the point $(R_0, R_1, \ldots, R_n)$). Let $d' =\deg_{x_n}(C')$. Then,
\begin{align*}
  0 &= C'(R_0,\ldots, R_n) + \sum_{i=1}^{d'} A^i \cdot \pfrac{\partial_{x_n^i}(C')(R_0,\ldots, R_n)}{i!} \bmod{\inangle{\vecz}^{j+1}}.
\end{align*}
Moreover, since $A = \Del_n(P_{n+j})(\vecz,\veca)$ is a homogeneous  polynomial of degree $j \geq 1$, we have $A^2 = 0\bmod{\inangle{\vecz}^{j+1}}$. Therefore,
\begin{align*}
  0 &= C'(R_0,\ldots, R_n) + \sum_{i=1}^{d'} A^i \cdot \pfrac{\partial_{x_n^i}(C')(R_0,\ldots, R_n)}{i!} \bmod{\inangle{\vecz}^{j+1}}\\
  & = C'(R_0,\ldots, R_n) + A \cdot \inparen{\partial_{x_n}(C')(R_0,\ldots, R_n)} \bmod{\inangle{\vecz}^{j+1}}\\
    & = C'(R_0,\ldots, R_n) + A \cdot \alpha \bmod{\inangle{\vecz}^{j+1}}
\end{align*}
where $\alpha = \partial_{x_n}(C')(R_0,\ldots, R_n)(\mathbf{0})$, the constant term of $\partial_{x_n}(C')(R_0,\ldots, R_n)(\vecz)$. 
Note
\begin{align*}
\alpha =   \partial_{x_n}(C')(R_0,\ldots, R_n)(\mathbf{0}) &= \partial_{x_n}(C')(\Del_0(P_{\leq n+j-1})(\mathbf{0},\veca),\ldots, \Del_n(P_{\leq n+j-1})(\mathbf{0},\veca))\\
& = \partial_{x_n}(C')(\Del_0(P_{\leq n+j-1})(\vecz,\veca),\ldots, \Del_n(P_{\leq n+j-1})(\vecz,\veca))(\mathbf{0})\\
                                                                    & = \partial_{x_n}(C')(\Del_0(P)(\vecz,\veca),\ldots, \Del_n(P)(\vecz,\veca))(\mathbf{0})\\
                                                           & = \inparen{\partial_{x_n}(C')\circ \g(P,\veca)}(\mathbf{0})\\
                                                           & = \Psi(\veca) \neq 0.
\end{align*}
Combining this with the previous equation, we get
\begin{align*}
  0 & = C'(R_0,\ldots, R_n) + A \cdot \Psi(\veca) \bmod{\inangle{\vecz}^{j+1}}\\
  \implies A = \Del_n(P_{n+j})(\vecz,\veca) & =\pfrac{-1}{\Psi(\veca)} \cdot C'(R_0,\ldots, R_n) \bmod{\inangle{\vecz}^{j+1}}.\hfill
\end{align*}
\qed \raisebox{0.2 em}{\scriptsize (\autoref{lem:inductive-reconstruction})}

\subsection{Derandomization from hard polynomial families}
\label{sec:derandomization-from-hard-polynomials}

In this section we will prove the more general form of our main result, \autoref{thm:derand-from-k-var-lbs-informal}.
It would be instructive to think of the case when $k$ is a constant, as it would retain the main message but simplify the theorem statement (for instance, the hitting set in the conclusion will be of size $s^{O(k^2/\delta^2)}$ instead of the more complicated expression in the statement below).

\begin{theorem}\label{thm:derand-from-k-var-lbs}
  Let $\delta > 0$ be an arbitrary constant, and $k:\N\rightarrow \N$ be any non-decreasing function with $k(d) = o\inparen{\sqrt{\log d}}$; let $\tilde{k}(d) = k\inparen{d^{\log d}}$.
  
Suppose $\set{P_{k(d),d}}_{d\in \N}$  is an explicit family\footnote{We are assuming that the polynomial family contains a degree $d$ polynomial, $P_{k,d}$, for every positive integer $d$. For the purposes of this theorem, it suffices to assume that the family is \emph{sufficiently often} in the following sense: There are absolute constants $a,b$ such that for any $t \in \N$, there is some $P_{k,d}$ in the family such that $t^a \leq d \leq t^b$.} of  polynomials such that $P_{k,d}$ is $k$-variate,  $\deg(P_{k,d}) = d$ and $P_{k,d}$ requires circuits of size at least $d^{\delta}$ (that is, $P_{k,d} \notin \mathcal{C}(k,d,d^\delta)$).
Then, there are explicit hitting sets of size $s^{O(\tilde{k}(s)^2)}$ for the class of $\mathcal{C}(s,s,s)$.
\end{theorem}
\begin{proof}
  We wish to construct a generator to fool the class of $s$-variate, degree $s$ polynomials computable by circuits of size $s$.   
  Let $t = \ceil{\frac{8}{\delta}}$. Choose the smallest $d$ such that $d > s^{(10t \cdot k(d) + 2) \cdot t}$. Note that this is always possible for some $d \leq s^{O(t^2 \log s)}$ since $k(d) = o(\sqrt{\log d})$ (furthermore, if $k = O(1)$, then we can in fact find a $d \leq s^{O(1)}$ satisfying the above constraint). We shall just use $k$ to denote $k(d)$. \\

  To the polynomial $P_{k,d}(z_1,\ldots, z_k)$ in the family, we associate the natural $kt$-variate polynomial $Q_{k,d,t}\inparen{\inparen{z_{i,j}\;:\; i\in [k]\;,\;j\in [t]}}$, of degree at most $d' = (kt) \cdot d^{1/t}$, such that $Q_{k,d,t}(z_{1,1},\ldots, z_{k,t})$ under the substitution
  \[
    z_{i,j} \rightarrow z_i^{d^{(j-1)/t}}
  \]
  yields the polynomial $P_{k,d}(z_1,\ldots, z_k)$. This is achieved by replacing each monomial $z_1^{e_1}\cdots z_k^{e_k}$ in $P_{k,d}$ by $\mathbf{z_{1,\ast}}^{\llparenthesis e_1\rrparenthesis} \cdots \mathbf{z_{k,\ast}}^{\llparenthesis e_k\rrparenthesis}$ where $\llparenthesis e_i \rrparenthesis$ is the tuple corresponding to the integer $e_i$ expressed in base $d^{1/t}$.

  Note that if $Q_{k,d,t}$ has a circuit of size $d'^{4} < d^{\delta/2 + o(1)}$ (recall that $k = o(\sqrt{\log d})$ and $t = O_\delta(1)$), then (by employing repeated squaring to perform the substitution described above) $P_{k,d}$ has a circuit of size at most
  \[
    d^{\delta/2  + o(1)} + O(kt \log d) \ll d^{\delta},
  \]
  which is a contradiction to the hardness of $P_{k,d}$. Hence we have that $Q_{k,d,t}$ is a $k$-variate polynomial of degree $d'$ that requires circuits of size at least $d'^4$. Since, by the choice of $d$, we have $d' > d^{1/t} > s^{10kt+2}$, we have
  \[
    d'^4> s^{10kt} \cdot s^2 \cdot d'^3. 
  \]
  By \autoref{thm:main}, we hence have that $\g_{Q_{k,d,t}}$ is a hitting-set generator for the class $\mathcal{C}(s,s,s)$. 
  Therefore, from \autoref{lem:Schwartz-Zippel}, this yields an explicit hitting set of size at most
  \begin{align*}
    (sd'+1)^{2t \cdot k} & =     (sd'+1)^{2t \cdot k(d)}\\
                         & \leq s^{O(k(d)^2)}\\
                         & \leq s^{O\inparen{k\inparen{s^{O(\log s)}}}^2}&\text{(since $d \leq s^{O(\log s)}$)}\\
                         & \leq s^{O\inparen{\tilde{k}(s)}^2} & \text{(since $k(s) = o(\sqrt{\log s}) \implies k(s^{O(\log s)}) = O(k(s^{\log s}))$)}
  \end{align*}
  as claimed by the theorem. 
\end{proof}



  

\paragraph{Comparison with the derandomization of Kabanets and Impagliazzo~\cite{KI04}:}
   The work of Kabanets and Impagliazzo~\cite{KI04} shows that if there is an explicit family of \emph{multilinear} polynomials $\set{Q_m}$ that requires circuits of size $2^{\Omega(m)}$, then we can construct explicit hitting sets of $s^{O(\log s)}$ size for $\mathcal{C}(s,s,s)$.

   From the above proof, it is easy to see that if $P_{k,d}$ requires circuits of size $d^{\delta}$, then the above proof shows that the polynomial $Q_{k,d,t}$ in the above proof, for $t = \log d$, is an  $m=O(k \log d)$-variate multilinear polynomial that requires circuits of size roughly $d^\delta = \exp(\Omega(m))$ to compute it. Therefore, the hypothesis of \autoref{thm:derand-from-k-var-lbs-informal} implies the hypothesis used in the result of Kabanets and Impagliazzo.

   However, the other direction seems unclear as it is conceivable that the circuit complexity of $P_{k,d}$ is significantly smaller than the complexity of $Q_{k,d,t}$. Hence, \autoref{thm:derand-from-k-var-lbs-informal}, when compared with the hardness-randomness trade-off of Kabanets and Impagliazzo~\cite{KI04}, can be interpreted as a potentially stronger hypothesis leading to the better derandomization of PIT (a complete derandomization of PIT in the case when $k = O(1)$).

\subsection{Application to bootstrapping phenomenon for hitting sets}
We now use~\autoref{thm:derand-from-k-var-lbs} to prove the following theorem about bootstrapping hitting sets for algebraic circuits. The following theorem shows that, for any constants $\delta, k > 0$,  any hitting set for $\mathcal{C}(k,\text{i-deg: }s, s^\delta)$, that saves \emph{even one point} from the trivial hitting set of size $(s+1)^k$ guaranteed by \autoref{lem:Comb-Null}, is sufficient to completely derandomize PIT for $k = O(1)$. 

\bssthm*

\begin{proof}
  Consider an arbitrary, large enough $s$ and let $H_s$ be the hitting set for $\mathcal{C}(k,\text{i-deg: }s,s^\delta)$; the hypothesis guarantees that $\abs{H_s} \leq (s+1)^{k} - 1 < (s+1)^k$. From \autoref{thm:HS}, we can then obtain a polynomial $P_s(z_1,\ldots, z_k)$ of individual degree at most $s$ that cannot be computed by circuits of size $s^\delta$. Expressing this in terms of its total degree $d \leq ks$, we get that $P_s$ is a degree $d$ polynomial that requires circuits of size more than
  \[
    \pfrac{d}{k}^{\delta} \gg d^{\delta/2}.
  \]
Hence, $\set{P_s}$ is an explicit family of $k$-variate, degree-$d$ polynomials that require circuits of size $d^{\delta/2}$. Therefore, \autoref{thm:derand-from-k-var-lbs} yields an explicit $\poly_{\delta,k}(s)$-sized hitting set for the class of $s$-variate, degree $s$ polynomials that can be computed by size $s$ circuits.   
\end{proof}





\medskip

Analogous statements also hold for slow-growing functions $k(d):\N\rightarrow \N$ yielding non-trivial derandomizations. However, the resulting hitting set and the time requried to construct it  would require super-polynomial (but subexponential) time. We state the above simpler settings as we believe they presents the main message without any notational clutter.

\subsection{The strengthened Shub-Smale $\tau$-conjecture implies PIT is in $\P$}
In this section, we prove that the strengthened $\tau$-conjecture (\autoref{conj:tau-conj-mod}) implies an efficiently enumeratable polynomial sized hitting set for algebraic circuits.  
\tautopit*
\begin{proof}
  Let $P_d(x) = (x-1)\cdots (x-d)$. If \autoref{conj:tau-conj-mod} is true with exponent $c$ as stated, then $P_d$ requires algebraic circuits of size $d^{1/c}$ to compute it. Let $f$ be a nonzero polynomial in $\mathcal{C}(s,s,s)$. Following the proof of \autoref{thm:derand-from-k-var-lbs}, for $t = O(c)$ and $d = \poly(s)$ chosen large enough, if $Q_{1,d,t}$ is the \emph{un-Kroneckered} version of $P_{d}$ used in the proof of  \autoref{thm:derand-from-k-var-lbs}, then $\g_{Q_{1,d,t}} \in (\F[y_1,\ldots, y_t, z_1,\ldots, z_t])^s$ is a hitting-set generator for $\mathcal{C}(s,s,s)$.

  Then, since $f \circ \g_{Q_{1,d,t}}$ is a $2t$-variate, nonzero polynomial of degree at most $sd$, by  \autoref{lem:Schwartz-Zippel} we have that the set
  \[
    \mathcal{H} = \setdef{\g_{Q_{1,d,t}}(\veca)}{\veca \in [sd+1]^{2t}}
  \]
  is a hitting set for $\mathcal{C}(s,s,s)$. As the coefficients of $P_d$ are bounded by $d^d \cdot d! = 2^{O(d\log d)}$, we have the same bound for the coefficients of $Q_{1,d,t}$. Using the definition of the generator (\autoref{def:prg}), all the coordinates of points in $\mathcal{H}$ are rational numbers with numerator and denominators being integers of $\poly(d,s) = \poly(s)$ bits and can certainly be computed by a Turing machine in $\poly(s)$ time. Thus, the above set $\mathcal{H}$ is a hitting set for $\mathcal{C}(s,s,s)$ that can be enumerated by a deterministic Turing machine in $\poly(s)$ time. 
\end{proof}

\section{Open Problems}
We end with some open problems. 
\begin{itemize}
\item The construction in this paper takes a $d^{\Omega(k)}$-hard, $k$-variate polynomial of degree $d$ to construct a generator for $\mathcal{C}(s,s,s)$ with $s \lessapprox d$. A $k$-variate degree $d$ polynomial has $d^{\Theta(k)}$ coefficients and if $P$ is $d^{\Omega(k)}$-hard then, intuitively, we have ``enough entropy'' to be able to fool $\mathcal{C}(s,s,s)$ with $s \approx d^{O(k)}$.

  On the other hand, the HSG of Kabanets and Impagliazzo~\cite{KI04}, in certain ranges of parameters, is ``lossless'' in the above sense it uses a $k$-variate, multilinear, hard polynomial which has $\tilde{s} = 2^k$ coefficients to build a generator for $\mathcal{C}(s,s,s)$ with $s = \tilde{s}^{\Omega(1)}$.

  In principle, it seems conceivable that there is a construction that, starting from a $d^{\Omega(k)}$-hard $k$-variate degree $d$ polynomial, gives an HSG for $\mathcal{C}(s,s,s)$ for $s = d^{\Omega(k)}$. Constructing such an explicit HSG is perhaps the most natural open question here. 
\item The construction of the HSG in this paper needs the characteristic of the field to be large enough or zero. Constructing a HSG with similar properties (seed length, stretch, running time, degree) over fields of small positive characteristic would be quite interesting. 
\item In the current statement of~\autoref{thm:main}, the hardness required for $P$ for the HSG to fool circuits of size $s$, depends also on the degree of this circuit.  We suspect that this dependence on the degree can be avoided, if the hardness assumption is in terms of \emph{border-complexity}; in that case, this HSG should fool all circuits of small size regardless of their degree. 
\item Lastly, it would be interesting to understand if this new HSG and the ideas in its analysis have any other applications in complexity theory. 
\end{itemize}

\section*{Acknowledgements}
We thank Marco Carmosino, Chi-Ning Chou, Nutan Limaye, Rahul Santhanam, Srikanth Srinivasan and Anamay Tengse for insightful conversations at various stages of this work. We are particularly grateful to Srikanth Srinivasan who pointed out an observation that strengthened our prior version of \autoref{thm:derand-from-k-var-lbs}, and eliminated the use of border complexity entirely.
We thank Ryan Williams and Chris Umans for pointing out the connection to the Tau conjecture. 
We also thank Madhu Sudan for many insightful discussions and much encouragement;  in particular for patiently sitting through a presentation of  a preliminary version of the proof of~\autoref{lem:inductive-reconstruction}. 

Mrinal is also thankful to Swastik Kopparty for many helpful discussions on  Nisan-Wigderson generator and list decoding algorithms for multiplicity codes while he was a Ph.D. student at Rutgers. Swastik's  interest in  a purely algebraic solution for algebraic hardness randomness tradeoffs had a non-trivial role in motivating this work.

\bibliographystyle{customurlbst/alphaurlpp}
\bibliography{references}

\end{document}